\documentclass[letterpaper, 10 pt, conference]{ieeeconf}
\IEEEoverridecommandlockouts 
\usepackage{max_paper}

\title{\LARGE \bf
Topological Obstructions to the Existence of Control Barrier Functions}

\author{Massimiliano de Sa and Aaron D. Ames%
\thanks{This research is supported by AFOSR Grant No. 113535-19668, Hybrid Dynamics - Deconstruction and Aggregation (HyDDRA).}% 
\thanks{The authors are with the Department of Control and Dynamical Systems, California Institute of Technology, Pasadena CA 91125, U.S.A.
        {\tt\small \{mdesa, ames\}@caltech.edu}}%
}

\begin{document}

\maketitle
\thispagestyle{empty}
\pagestyle{empty}

\renewcommand{\bf}{\mathbf{f}}

%%%%%%%%%%%%%%%%%%%%%%%%%%%%%%%%%%%%%%%%%%%%%%%%%%%%%%%%%%%%%%%%%%%%%%%%%%%%%%%%
\begin{abstract}

In 1983, Brockett developed a topological necessary condition for the existence of continuous, asymptotically stabilizing control laws. Building upon recent work on necessary conditions for set stabilization, we develop Brockett-like necessary conditions for the existence of control barrier functions (CBFs). By leveraging the unique geometry of CBF safe sets, we provide simple and self-contained derivations of necessary conditions for the existence of CBFs and their safe, continuous controllers. We demonstrate the application of these conditions to instructive examples and kinematic nonholonomic systems, and discuss their relationship to Brockett's necessary condition.

\end{abstract}

\section{Introduction}
\label{sec:introduction}
In recent years, control barrier functions (CBFs) have emerged as a leading framework for the design of safety-critical controllers \cite{ames2016control}. Herein, a safety specification for a control system is encoded as the superlevel set of a function, and safety is achieved by rendering this set forward invariant.
In the time since the introduction of CBFs, the community has undertaken a thorough study of the synthesis of CBFs and their associated safety-critical controllers \cite{ames2019control}. Further, the dynamical properties of control systems with CBF-based safety-critical controllers have been extensively studied \cite{mestres2025control, reis2020control}. Though the problem of discerning closed-loop dynamical properties from a CBF and a controller has been well-studied, the converse problem---studying the \textit{existence} of a CBF and a safety-critical controller using their closed-loop dynamical properties---has been comparatively overlooked. This problem finds valuable applications in the automated synthesis of safety certificates \cite{dawson2023safe}, where it can inform whether a CBF exists for a proposed safe set specification.

In the now-classic work of Brockett, the analogous converse problem for the asymptotic stabilization of equilibria is treated \cite{brockett1983asymptotic}. By studying the topological and dynamical properties of an asymptotically stable system, Brockett derived necessary conditions that a control system asymptotically stabilizable by continuous feedback must satisfy. In doing so, he posed an elegant test for the continuous stabilizability of a broad class of control systems. In the context of CBFs, the power of Brockett's necessary condition begs the question: do there exist Brockett-like necessary conditions for safety?

Recently, in the work of Kvalheim and Koditschek, this question was answered in the affirmative \cite{kvalheim2021necessary}. Using cohomological tools, it was shown that Brockett-like necessary conditions can be derived for the stabilization and strict forward invariance of a class of compact sets. 
\begin{figure}
    \centering
    \includegraphics[width=0.95\linewidth]{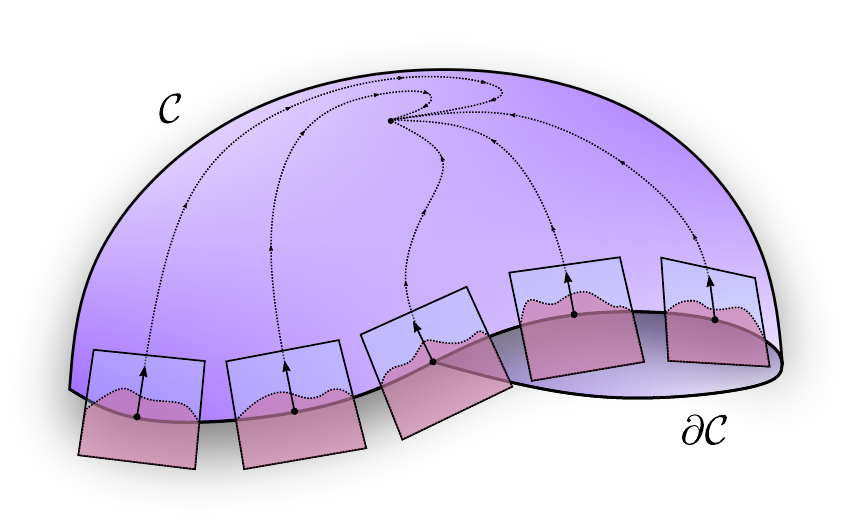}
    \caption{For a safe set $\C$ of a certain geometry and a safe vector field $X$, small \textit{perturbations} $Z$ (in red) of $X$ satisfy $X_p - Z_p = 0_p$ for some point $p$. This observation can be used to derive necessary conditions for safety.}
    \label{fig:cbf-brockett}
    \vspace{-10pt}
\end{figure}

Although the topological and dynamical machinery of \cite{kvalheim2021necessary} leads to general and powerful necessary conditions, it carries with it several drawbacks for the specific setting of control barrier functions. Compared to the potentially pathological safe sets of \cite{kvalheim2021necessary}, the safe sets studied in the CBF framework carry a substantial geometric structure. This suggests that a more elementary, geometric path to deriving necessary conditions for safety is possible within the CBF framework. The tools used by \cite{kvalheim2021necessary} also necessitate the assumption of unique integrability of the closed-loop system, a potentially unnecessary requirement for the CBF case. Further, in addition to the strict form of forward-invariance considered by \cite{kvalheim2021necessary}, the non-strict case is also of interest in the CBF literature.

These points raise the following questions: by focusing on the specific setting of control barrier functions, can one chart a more elementary path to deriving necessary conditions for safety? Further, can the structure of the CBF framework be utilized to relax the unique integrability assumption of \cite{kvalheim2021necessary} and treat the case of non-strict forward invariance?

In this work, we shall utilize the unique geometric structure of CBFs to provide an elementary, geometric approach to deriving necessary conditions for safety. Using this geometric structure, we shall relax the unique integrability assumption of \cite{kvalheim2021necessary} and treat both the non-strict and strict cases of forward invariance. By developing necessary conditions for safety through this CBF-theoretic lens, we hope that these conditions and their derivations will be both accessible and useful to the broad CBF community. 

The remainder of the paper is structured as follows. In Section II, we review the background on geometry, CBFs and Brockett's necessary condition. In Section III, we develop a result on the existence of zeros of vector fields. In Section IV, we use this result to derive Brockett-like necessary conditions for CBFs. Here, we also demonstrate their application to practical systems and discuss their relation to Brockett's condition. Finally, in Section V, we give concluding remarks.

\section{Preliminaries}\label{sec:prelims}
To begin, we review the essential background on geometry, CBFs, and Brockett's condition. To simplify our exposition, we work with \textit{smooth} $(C^\infty)$ CBFs and manifolds throughout.

\subsection{Differential Geometry}
First, we informally recall the key concepts of differential geometry, deferring precise definitions to \cite{lee2012intro}. A \textit{smooth manifold with boundary} is a space $\M$ for which every point has a neighborhood $U$ diffeomorphic to an open subset of $\R^n$ or $\h^n = \{x \in \R^n : x^n \geq 0\}$. A pair $(U, \phi)$ of such a neighborhood and diffeomorphism is called a \textit{local coordinate chart}. The \textit{boundary} $\partial \M$ of $\M$ consists of all $p \in \M$ mapped by some $(U, \phi)$ to an $x \in \h^n$ with $x^n = 0$. The \textit{interior} $\Int \M$ is the complement of $\partial \M$ in $\M$. If $\partial \M = \emptyset$, $\M$ is simply called a \textit{smooth manifold}. The \textit{dimension} $\dim(\M)$ of $\M$ is the dimension $n$ of $\R^n$ or $\h^n$; we will use $\dim(\M) = n$ throughout.

A tangent vector to $\M$ at $p$ is written $v_p$, the tangent space to $\M$ at $p$ as $T_p \M$ and the tangent bundle as $T\M$. For a smooth map $F: \M_1 \to \M_2$ between manifolds, the differential of $F$ at $p$ (analogous to the Jacobian in $\R^n$) is a linear map $dF_p : T_p \M_1 \to T_{F(p)} \M_2$. The set of all smooth functions $h: \M \to \R$ is written $C^\infty(\M)$; the differential $dh_p$ of $h$ at $p$ is interpreted as a linear map from $T_p \M \to \R$.

A \textit{vector field} on $\M$ is a map $X: \M \to T\M$ assigning each $p \in \M$ to a tangent vector $v_p \in T_p \M$. We write the value of $X$ at $p$ as $X_p$ or $X|_p$. We write $\X^0(\M)$ and $\X(\M)$ for the sets of all continuous and smooth vector fields on $\M$. For any local coordinate chart $(U, \phi)$, there is a set of smooth \textit{coordinate vector fields} $\{\frac{\partial}{\partial x^i}\}_{i = 1}^n$ defined on $U$.

A vector field $X \in \X^0(\M)$ is \textit{uniquely integrable} if it has a unique maximal flow $\varphi$. This is a continuous map $\varphi: D \subseteq \R \times \M \to \M$, where each slice $D^{(p)} = \{t \in \R : (t, p) \in D\}$ is the maximal interval of existence at $p$, satisfying $\tfrac{d}{d\tau}\big|_{\tau = t}\varphi_{\tau}(p) = X_{\varphi_{t}(p)}$ and $\varphi_0(p) = p$ for all $(t, p) \in D$. We write $D^{(p)}_{\geq 0}$ for $D^{(p)} \cap \R_{\geq 0}$ and $D^{(p)}_{>0}$ for $D^{(p)} \cap \R_{>0}$. Given a function $h \in C^\infty(\M)$, its rate of change along the flow of $X$ is computed $\frac{d}{dt}\big|_{t = 0} h(\varphi_t(p)) = dh_p X_p$.

A \textit{Riemannian metric} on $\M$ is a smooth assignment of each $p \in \M$ to an inner product $\braket{\cdot, \cdot}$ on $T_p \M$, which induces a norm $\norm{\cdot}$ on $T_p \M$. The \textit{gradient} of $h \in C^\infty(\M)$ with respect to a metric is the smooth vector field $p \mapsto \grad h|_p$ satisfying $\braket{\grad h|_p, v_p} = dh_p v_p$ across $T\M$.

\subsection{Nonlinear Control Systems}
\begin{defn}\label{defn:nl-sys}
    A \textit{nonlinear control system} is a tuple $\Sigma = (\U, \M, F)$ of a set $\U$ (the \textit{input space}), a smooth manifold $\M$ (the \textit{state space}), and a map $F: \M \times \U \to T\M$ (the \textit{dynamics}), taking $(p, u) \in \M \times \U$ to vectors $v_p\in T_p \M$.
\end{defn}
\begin{example}\label{ex:eucl-nl-sys}
    Let $f: \R^n \times \U \to \R^n$, for $\U \subseteq \R^m$. The Euclidean control system $\dot x = f(x, u)$ induces a nonlinear control system $\Sigma = (\U, \R^n, F)$ in the sense of Definition \ref{defn:nl-sys}. Recalling that $T\R^n$ is diffeomorphic to $\R^n \times \R^n$, the dynamics $F$ of $\Sigma$ are specified by $F: (x, u) \mapsto (x, f(x, u))$.
\end{example}
The equations of motion of a system $\Sigma = (\U, \M, F)$ are written $\dot p = F(p, u)$. The dynamics of $\Sigma$ under a controller $\kappa: \M \to \U$ are $\dot p = F(p, \kappa(p))$. As $p \mapsto F(p, \kappa(p))$ assigns each $p \in \M$ to a $\dot p = F(p, \kappa(p)) \in T_p \M$, it is a vector field on $\M$; we call this the \textit{closed-loop vector field}.

\begin{defn}\label{defn:contr-aff-sys}
    A \textit{control-affine system} is a tuple $\Sigma = (\M, X_0,..., X_m)$ of a smooth manifold $\M$, a \textit{drift vector field} $X_0$ on $\M$, and \textit{input vector fields} $X_1, ..., X_m$ on $\M$. $\Sigma$ is \textit{continuous/locally Lipschitz} if each $X_i$ is so.
\end{defn}

A control-affine system $\Sigma = (\M, X_0,..., X_m)$ induces a nonlinear control system $\tilde \Sigma = (\U, \M, F)$ with $\U = \R^m$ and $F: \M \times \R^m \to T\M$ the map $(p, u) \mapsto X_0|_p + \sum_{i = 1}^m u^i X_i|_p $, for each $u^i \in \R$ a component of $u = (u^1, ..., u^m) \in \R^m$.

\subsection{Control Barrier Functions \& Safe Set Geometry}
Let $\Sigma = (\U, \M, F)$ be a control system and $\C \subseteq \M$ a prescribed \textit{safe set} of states in which $\Sigma$ must remain. The core aim of safety-critical control is to design a controller $\kappa$ that renders $\C$ \textit{forward invariant} for the closed-loop system.
\begin{defn}\label{defn:fwd-inv}
    Let $X \in \X^0(\M)$ be uniquely integrable. A subset $\C \subseteq \M$ is said to be \textit{forward invariant} with respect to $X$ if $p \in \C$ implies $\varphi_t(p) \in \C, \; \forall t \in D^{(p)}_{\geq 0}$.
\end{defn}
In the \textit{control barrier function (CBF)} framework, one considers safe sets of the form:
\begin{align}
    \C = \{p \in \M : h(p) \geq 0\} = h^{-1}(\R_{\geq 0}) \label{eqn:safe-set},
\end{align}
where $h : \M \to \R$ is a smooth function with zero a regular value ($dh_p \neq 0$ whenever $h(p) = 0$). Such a set is a \textit{regular domain} in $\M$: a properly embedded submanifold of $\M$ with boundary for which $\dim \C = \dim \M$ \cite[Prop. 5.47]{lee2012intro}.

The forward invariance of a regular domain $\C$ with respect to a vector field $X$ is characterized by the behavior of $X$ at $\partial \C$. For $p \in \partial \C$, a vector $v_p \in T_p \C \setminus T_p \partial \C$ is \textit{inward-pointing} at $\partial \C$ if there is a curve $\gamma: [0, \epsilon) \to \C$ with $\gamma(0) = p$ and $\gamma'(0) = v_p$. Similarly, a vector $w_p$ is \textit{outward-pointing} at $\partial \C$ if $-w_p$ is inward-pointing. For the regular domain \eqref{eqn:safe-set}, a vector $v_p$ is inward-pointing at $\partial \C$ if and only if $dh_p v_p > 0$, and is tangent to $\partial \C$ if and only if $dh_p v_p = 0$.
\begin{theo}\label{theo:geom-nagumo}
    Let $X \in \X^0(\M)$ be uniquely integrable with flow $\varphi$ and $\C$ a regular domain in $\M$.
    \begin{enumerate}
        \item If for each $ p \in \partial \C$, $X_p$ is inward-pointing or tangent to $\partial \C$, then $\C$ is forward invariant with respect to $X$.
        \item If for each $p \in \partial \C$, $X_p$ is inward-pointing at $\partial \C$, then for all $p \in \partial \C$ and $t \in D^{(p)}_{> 0}$, $\varphi_t(p) \in \Int \C$.
    \end{enumerate}
\end{theo}
\begin{proof}
    Item (1) follows directly from \cite[Thm. 1]{hartman1972invariant} and \cite[Thm. 1(i)]{de2025bundles}, while item (2) is remarked by \cite[p. 669]{kvalheim2021necessary}.    
\end{proof}

For the set \eqref{eqn:safe-set}, Theorem \ref{theo:geom-nagumo} gives a non-strict condition $dh_p X_p \geq 0$ on $\partial \C$ for forward invariance, and a strict condition $dh_p X_p > 0$ on $\partial \C$ for \textit{strict} forward invariance. These conditions motivate the following definition. 
\begin{defn}\label{defn:cbf}
    Let $\Sigma = (\U, \M, F)$ be a nonlinear control system. A function $h \in C^\infty(\M)$ is a \textit{control barrier function (CBF)} for $\Sigma$ with safe set $\C = h^{-1}(\R_{\geq 0})$ if it has zero as a regular value and there exists a function\footnote{A function $\alpha: \R \to \R$ belongs to the class $\K^e_\infty$ if it is continuous, strictly increasing, and satisfies $\alpha(0) = 0$ and $\; \lim_{r \to \pm \infty}\alpha(r) = \pm \infty$.} $\alpha \in \K^e_\infty$ for which
    \begin{align}
        \sup_{u \in \U} dh_p F(p, u) \geq -\alpha(h(p)), \; \forall p \in \C. \label{eq:cbf}
    \end{align}
    If \eqref{eq:cbf} holds on a neighborhood $\D$ of $\C$, we call $h$ a \textit{CBF on $\D$}. If \eqref{eq:cbf} holds with \textit{strict} inequality, we call $h$ a \textit{strict CBF}.
\end{defn}
\begin{example}
    For the Euclidean system $\dot x = f(x, u)$ of Example \ref{ex:eucl-nl-sys}, \eqref{eq:cbf} is equivalent to the ``classic" CBF condition $\sup_{u \in \U} \frac{\partial h}{\partial x} f(x, u) \geq -\alpha(h(x)), \; \forall x \in \C$ of \cite{ames2016control}.
\end{example}
Since $h(p) = 0$ on $\partial \C$, $\alpha(h(p)) = 0$ on $\partial \C$ as well. Thus, \eqref{eq:cbf} reduces to $\sup_{u \in \U} dh_p F(p, u) \geq 0$ on $\partial \C$, resembling the forward invariance condition of Theorem \ref{theo:geom-nagumo}. Now, we aim to use condition \eqref{eq:cbf} to obtain a controller $\kappa$ for which
\begin{align}
    dh_p F(p, \kappa(p)) \geq -\alpha(h(p)), \; \forall p \in \C. \label{eq:cl-loop-cbf}
\end{align}
\begin{prop}\label{prop:strict-aff}
    Let $\Sigma = (\M, X_0, ..., X_m)$ be a continuous control-affine system and $h$ a strict CBF for $\Sigma$ with $\alpha$ and $\C$ as in \eqref{eq:cbf}. Then, there exists a smooth control law $\kappa: \M \to \R^m$ and a neighborhood $\D$ of $\C$ for which
    \begin{align}
        dh_p\Big[X_0|_p + \sum_{i = 1}^m \kappa^i(p) X_i|_p\Big] > -\alpha(h(p)), \; \forall p \in \D. \label{eq:cbf-contr-ineq}
    \end{align}
\end{prop}
\begin{remark}
    If $h$ is instead a non-strict CBF, there may be no continuous controller for which \eqref{eq:cl-loop-cbf} is satisfied \cite{alyaseen2025continuity}. 
\end{remark}
\begin{proof}
    We follow an Artstein-type argument \cite{artstein1983stabilization}. Suppose that $h$ is a strict CBF for $\Sigma$. Fix $q \in \C$. Since $\Sigma$ is continuous, there is a neighborhood $V_q$ of $q$ and an input $u_q \in \R^m$ for which $dh_p [X_0|_p + \sum_{i = 1}^m u_q^i X_i|_p ] > -\alpha(h(p))$ for each $p \in V_q$. Let $\D' \defeq \bigcup_{q \in \C} V_q$, and take a locally finite partition of unity $\{\psi_q\}_{q \in \C}$ subordinate to $\{V_q\}_{q \in \C}$ with domain $\D'$ \cite[p. 43]{lee2012intro}.
    On $\D'$, define a controller $\kappa': p \mapsto \sum_{q \in \C} u_q \psi_q(p)$. Substituting $\kappa'(p)$ for $u$, we find that for each $p \in \D'$,
    \begin{align}
        &dh_p\Big[X_0|_p + \sum_{i = 1}^m X_i|_p \sum_{q \in \C} u_q^i \psi_q(p) \Big]\\
        &=  \sum_{q \in \C} \psi_q(p) dh_p \Big[X_0|_p + \sum_{i = 1}^m u_q^i X_i|_p \Big] > -\alpha(h(p)).
    \end{align}
    Let $\D$ be a neighborhood of $\C$ with $\overline \D \subseteq \D'$, and take $\kappa$ as any smooth extension of $\kappa'$ to $\M$ with $\kappa|_{\overline \D} = \kappa'|_{\overline \D}$.
\end{proof}

\subsection{Brockett's Necessary Condition}
Finally, we recall \textit{Brockett's necessary condition} \cite[Thm. 1.(iii)]{brockett1983asymptotic}. Though Brockett's original result was written for $C^1$ control systems and feedback laws, it is known to hold in the following strengthened form \cite[p. 660]{kvalheim2021necessary}.
\begin{theo}[Brockett]\label{theo:brockett}
    Consider a Euclidean control system $\dot x = f(x, u)$ on $\R^n \times \R^m$ and a point $x^* \in \R^n$. Suppose there is a $\kappa: \R^n \to \R^m$ for which the closed-loop $x \mapsto f(x, \kappa(x))$ is continuous, uniquely integrable, and locally asymptotically stable to $x^*$. Then, there is a neighborhood $\W$ of $0$ in $\R^n$ such that for all $z \in \W$, there exist $(x, u) \in \R^n \times \R^m$ with
    \begin{align}
        f(x, u) = z.
    \end{align}
\end{theo}
\begin{example}\label{ex:brockett-nh}
    Consider the \textit{nonholonomic integrator} \cite[p. 181]{brockett1983asymptotic} on $\R^3 \times \R^2$, specified by $(\dot x^1, \dot x^2, \dot x^3) = (u^1, u^2, x^2 u^1 - x^1u^2)$. For each nonzero $\epsilon$, there are no $(x, u) \in \R^3 \times \R^2$ for which $(0, 0, \epsilon) = f(x, u)$. As such, there is no neighborhood $\W$ of $0$ for which $f(x, u) = z$ has a solution $(x, u)$ for each $z \in \W$. Thus, Brockett's condition is not satisfied.
\end{example}

\section{The Zeros of Vector Fields}
At a high level, given a system $\Sigma$ and a proposed regular domain safe set $\C$, we wish to determine if there is a CBF and corresponding controller for $\C$ for which the closed-loop vector field is safe and continuous. We will approach this problem by proposing necessary conditions for each CBF variant of Definition \ref{defn:cbf}.

Similar to Brockett and Kvalheim \& Koditschek, we shall construct these necessary conditions from topological results on the zeros of vector fields. To this end, we now establish a relationship between safe set geometry and the zeros of safe vector fields.

\subsection{Ideas from Topology}
We begin by surveying several elementary yet important topological ideas. For precise definitions, we refer to \cite{hatcher2005algebraic}.

\subsubsection{Homotopy} Two continuous maps $f, g: X \to Y$ are \textit{homotopic} if there is a continuous $F: [0, 1] \times X \to Y$ with $F(0, x) = f(x)$ and $F(1, x) = g(x)$ for all $x \in X$. A map $f: X \to Y$ is a \textit{homotopy equivalence} if there is a map $g: Y \to X$ for which $f \circ g$ is homotopic to $\id_Y$ and $g \circ f$ to $\id_X$. If there exists a homotopy equivalence $f:X \to Y$, then $X$ and $Y$ are \textit{homotopy equivalent} spaces; intuitively, such $X$ and $Y$ are continuously deformable into one another.

\subsubsection{CW Complexes} \textit{CW complexes} are a rich class of spaces formed by gluing together simple building blocks called $n$-cells, sets of the form $\{x \in \R^n : \norm{x}_{\ell^2} \leq 1\}$, in an inductive procedure. We sketch this procedure below. A \textit{finite CW complex} $X$ is constructed:
\begin{figure}
    \centering
    \includegraphics[width=0.9\linewidth]{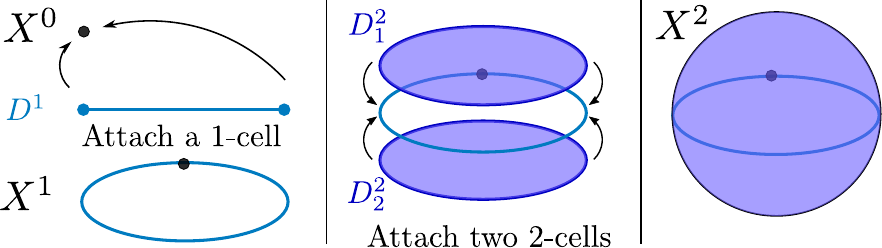}
    \caption{We construct $\sph^2$ as a finite CW complex. First, construct an equator by gluing the boundary of a 1-cell to a $0$-cell. Then, glue the boundaries of two 2-cells to the equator to form $\sph^2$. This construction shows $\chi(\sph^2) = 2$.}
    \label{fig:cw-complex}
\end{figure}
\begin{enumerate}
    \item Begin with the $0$-cells, a finite set of points $X^0$.
    \item Let $\{D^n_\alpha\}$ be a finite set of $n$-cells. Glue the boundary of each $D^n_\alpha$ to $X^{n-1}$ and name the resulting space $X^n$.
    \item Terminate at a finite $n$ and define $X \defeq \bigcup_{i = 0}^n X^i$.
\end{enumerate}
In Fig. \ref{fig:cw-complex}, we realize the 2-sphere as a finite CW complex. Many familiar spaces are realizable as finite CW complexes; notably, any compact manifold with boundary is homotopy equivalent to a finite CW complex \cite[p. 166]{hirsch1976differential}.

\subsubsection{Euler Characteristic}
The structure of CW complexes is conducive to defining simple objects that are invariant under homotopy equivalences. For a finite CW complex $X$, one such invariant is the \textit{Euler characteristic}, $\chi(X) \defeq \sum_{k = 0}^n (-1)^k (\text{\# of $k$-cells in $X$})$, where $n$ is the largest dimension of cell in $X$. Any two homotopy equivalent finite CW complexes share the \textit{same} Euler characteristic. As such, the Euler characteristic of a compact manifold with boundary may be taken as the Euler characteristic of any finite CW complex to which it is homotopy equivalent. 

\subsection{Application to Zeros of Vector Fields}
Now, we present a relationship between the Euler characteristic of a regular domain and the zeros of its vector fields.
\begin{prop}\label{prop:ec-zeros}
    Let $X \in \X^0(\M)$ and $\C$ be a compact regular domain in $\M$ with $\chi(\C) \neq 0$. If for all $p \in \partial \C$, $X_p$ is inward-pointing or tangent to $\partial \C$, then $X_p = 0_p$ for some $p \in \C$.
\end{prop}
\begin{proof}
    First, suppose $X$ is inward-pointing at $\partial \C$. It is a direct consequence of the continuous Poincar\'e-Hopf theorem \cite[p. 35]{milnor1997topology} that $X$ has a zero in $\C$ \cite[Lem. 2.2]{kvalheim2021necessary}.
    
    Now, suppose $X$ points inward or tangent to $\partial \C$. There is a $Y \in \X(\M)$ that points inward at $\partial\C$ \cite[p. 200]{lee2012intro}. For each $\delta > 0$, $X + \delta Y$ points inward at $\partial \C$. Taking a sequence $\delta_n = \frac{1}{n}$, $n \in \mathbb N$, we use the inward-pointing case to obtain a sequence of points $\{p_n\} \subseteq \C$ for which $X_{p_n} + \delta_n Y_{p_n} = 0_{p_n}$ for all $n \in \mathbb N$. Since $\C$ is compact, $\{p_n\}$ has a convergent subsequence $\{p_{n_k}\}$ with limit $p \in \C$. Passing to the limit, we conclude $X_p = \lim_{k \to \infty} X_{p_{n_k}} + \delta_{n_k} Y_{p_{n_k}} = 0_p$.
\end{proof}

\section{Topological Obstructions for CBFs}
Above, we showed that, under the appropriate topological conditions, a vector field that points inward or tangent at the boundary of a regular domain has a zero within the regular domain. As CBF safe sets are regular domains, and CBF-based controllers enforce an inward-pointing or tangency condition, the existence of zeros guaranteed by Proposition \ref{prop:ec-zeros} can be fashioned into a necessary condition for CBFs. 

We accomplish this by the following approach, which is directly parallel to that of \cite{brockett1983asymptotic} and \cite{kvalheim2021necessary}. Given a system $\Sigma$ and a proposed CBF safe set $\C$ with $\chi(\C) \neq 0$, execute:
\begin{enumerate}
    \item \underline{Controller}: suppose $\exists \; \kappa$ for which the closed-loop $p \mapsto F(p, \kappa(p))$ is continuous \& satisfies a CBF condition.
    \item \underline{Perturbation}: pick $Z \in \X^0(\M)$ for which $F(p, \kappa(p)) - Z_p$ is inward-pointing or tangent to $\partial \C$.
    \item \underline{Zero}: conclude $\exists \; p$ for which $F(p, \kappa(p)) - Z_p = 0_p$.
    \item \underline{Input}: conclude $\exists \; (p, u)$ for which $F(p, u) = Z_p$.
\end{enumerate}

This last step produces a necessary condition for safety analogous to Brockett's original necessary condition. This argument---which also powers the necessary conditions of Brockett and Kvalheim \& Koditschek---will also underlie our conditions. We apply this argument in three cases, starting from the weakest assumptions and ending with the strongest.

\subsection{Non-strict CBF on the Safe Set}
First, we treat the case of a non-strict CBF on $\C$. Consider $\Sigma = (\U, \M, F)$ with a CBF $h$ and a controller $\kappa$ for which
\begin{align}
    dh_p F(p, \kappa(p)) \geq -\alpha(h(p)), \; \forall p \in \C, \label{eq:nonstr-C}
\end{align}
where $\alpha \in \K_\infty^e$. We now implement the four-step procedure above to obtain necessary conditions for this setting.
\begin{theo}\label{theo:nonstr-cbf-c}
    Let $h$ be a CBF for $\Sigma = (\U, \M, F)$ with compact $\C$ and $\chi(\C) \neq 0$. Let $\kappa$ be a controller for which $p \mapsto F(p, \kappa(p))$ is continuous and satisfies \eqref{eq:nonstr-C}. Then, for any $Z \in \X^0(\M)$ with $dh_p Z_p \leq 0$ on $\partial \C$, there exist $(p, u) \in \C \times \U$ for which 
    \begin{align}
        F(p, u) = Z_p.
    \end{align}
\end{theo}
\begin{proof}
    Let $\kappa, h, Z$ be as in the statement of Theorem \ref{theo:nonstr-cbf-c}. Since $dh_p F(p, \kappa(p)) \geq 0$ and $dh_p Z_p \leq 0$ on $\partial \C$, it must be that $dh_p[F(p, \kappa(p)) - Z_p] \geq 0$ on $\partial \C$. Thus, for each $p \in \partial \C$, the continuous vector field $F(p, \kappa(p)) - Z_p$ is inward-pointing or tangent to $\partial \C$. By Proposition \ref{prop:ec-zeros}, there is a $p \in \C$ for which $F(p, \kappa(p)) - Z_p = 0_p$. We conclude the existence of a state-input pair $(p, u) \in \C \times \U$ for which $F(p, u) = Z_p$. 
\end{proof}

\subsection{Non-strict CBF on a Neighborhood of the Safe Set}
Next, we consider the case of a nonlinear control system $\Sigma = (\U, \M, F)$ with a CBF $h$ and a controller $\kappa$ for which
\begin{align}
    dh_p F(p, \kappa(p)) \geq -\alpha(h(p)), \; \forall p \in \D, \label{eq:nonstr-cbf-neigh}
\end{align}
where $\D$ is a neighborhood of $\C$ and $\alpha \in \K_\infty^e$. We begin by proving a technical lemma that adapts a classic result from Morse theory \cite[Thm. 3.1]{milnor1963morse} to our setting.
\begin{figure}
    \centering
    \includegraphics[width=0.85\linewidth]{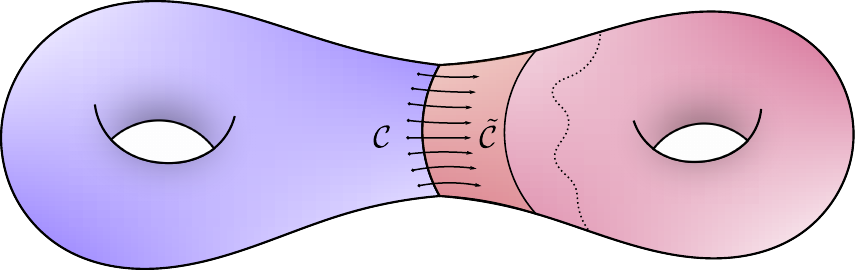}
    \caption{By \textit{flowing out} from $\C$, we can inflate $\C$ to a new set $\tilde \C$ diffeomorphic to $\C$, on which a safe vector field for $\C$ is inward-pointing.}
    \label{fig:flowout}
\end{figure}
\begin{lem}\label{lem:c-perturb}
    Let $\M$ be a smooth manifold and $h \in C^\infty(\M)$ have zero a regular value and a compact $0$-superlevel set $\C$. Suppose there is an $\alpha \in \K_\infty^e$ for which $X \in \X^0(\M)$ satisfies $dh_p X_p \geq -\alpha(h(p))$ on a neighborhood $\D$ of $\C$. Then, there is a compact regular domain $\tilde \C \subseteq \D$ for which:
    \begin{enumerate}
        \item $\C \subseteq \tilde \C$ and $\C$ is diffeomorphic to $\tilde \C$.
        \item $X$ is inward-pointing at $\partial \tilde \C$.
    \end{enumerate}
\end{lem}
\begin{proof}
    The key idea is depicted in Fig. \ref{fig:flowout}: by \textit{flowing out} from $\C$ via the flow of an appropriately-constructed, outward-pointing vector field $Y$, we can inflate $\C$ to the desired $\tilde \C$.

    First, we construct the vector field $Y$. Fix any Riemannian metric on $\M$. By shrinking $\D$ if necessary, we assume without loss of generality that it has compact closure. Since $dh_p X_p \geq -\alpha(h(p))$ for all $p \in \D$, it follows from continuity of $p \mapsto dh_p X_p$ that $dh_p X_p \geq -\alpha(h(p))$ for all $p \in \overline \D$. Since zero is a regular value of $h$ and $\alpha(h(p)) < 0$ on $\overline \D \setminus \C$, we conclude that $h$ has no critical points in $\overline \D \setminus \Int \C$. Thus,
    \begin{align}
        p \mapsto - \tfrac{\grad h|_p}{\norm{\grad h|_p}^2}, \; p \in \overline \D \setminus \Int \C, \label{eq:Y-defn}
    \end{align}
    is smooth and well-defined on $\overline \D \setminus \Int \C$. Since $\overline \D \setminus \Int \C$ is compact, we may extend \eqref{eq:Y-defn} to a smooth vector field $Y \in \X(\M)$ that vanishes outside of a compact set \cite[Lem. 8.6]{lee2012intro}. Since $Y$ is zero outside of a compact set, it has a smooth flow $\varphi_t: \M \to \M$ defined for all $t \in \R$ \cite[Thm. 9.16]{lee2012intro}. By construction, $Y$ points outward at the boundary of $\C$. As such, for any $p \in \C$ and $t < 0$, $\varphi_t(p) \in \Int \C$. Similarly, for any $p \in \partial \C$ and $t > 0$, $\varphi_t(p) \in \M \setminus \C$. Thus, as depicted by Figure \ref{fig:flowout}, the positive flow of $Y$ expands $\C$ into $\M \setminus \C$. 

    Now, we characterize the image of $\C$ under the flow of $Y$ for small positive times. By continuity of $\varphi$ and compactness of $\C$, there is a $t_1 > 0$ for which $\varphi_{t}(\C) \subseteq \D$ $\forall t \in [0, t_1]$. Since $Y$ is outward-pointing, for any $t \in [0, t_1]$ and $p \in \partial \C$, $\varphi_t(p) \in \D \setminus \Int \C$. Thus, for any $p \in \partial \C$ and $t \in [0, t_1]$,
    \begin{align}
        \tfrac{d}{d\tau}\big|_{\tau = t} h(\varphi_{\tau}(p)) &= - \tfrac{\braket{\grad h|_{\varphi_t(p)}, \grad h|_{\varphi_t(p)}}}{\norm{\grad h|_{\varphi_t(p)}}^2} = -1.
    \end{align}
    We conclude that $h(\varphi_{t}(p)) = -t$ for any $t \in [0, t_1]$ and $p \in \partial \C$. In context of Figure \ref{fig:flowout}, this suggests a connection between $\varphi_t(\C)$ and the $-t$-superlevel set of $h$. By making this connection explicit, we shall arrive at the desired $\tilde \C$. 
    
    For each $t \in [0, t_1]$, define $\C_t \defeq \varphi_{t}(\C)$. Each $\varphi_t$ is a diffeomorphism, and therefore preserves the boundary and interior of $\C$. Thus, for any $t \in [0, t_1]$, $\partial \C_t = \varphi_{t}(\partial \C)$ and $\Int \C_t = \varphi_t(\Int \C)$. Further, since $\varphi_t$ is a diffeomorphism, each $\C_t$ is a compact regular domain. Now, define $h_1: \Int \C_{t_1} \to \R$ as the restriction of $h$ to $\Int \C_{t_1}$ and fix a time $t_0 \in (0, t_1)$. Following \cite[Thm. 3.1]{milnor1963morse}, $\C_{t_0} = h_1^{-1}(\R_{\geq -t_0})$. 
    
    Define $\tilde \C \defeq \C_{t_0}$. This is a compact regular domain contained in $\D$, containing $\C$, and diffeomorphic to $\C$. Since $\tilde \C$ is an $h_1$-superlevel set and $h_1 = h|_{\Int \C_{t_1}}$, any vector $v_p \in T_p \M$ with $p \in \partial \tilde \C$ satisfying $dh_p v_p > 0$ is inward-pointing at $\partial \tilde \C$. As $dh_p X_p \geq -\alpha(h(p))$ and $\alpha(h(p)) < 0$ on $\partial \tilde \C$, we conclude that $X$ points inward at $\partial \tilde \C$. 
\end{proof}

With this result in hand, we are ready to derive our next condition. To state this condition in topological language, we require the following additional concept \cite{kvalheim2021necessary}. Write $0_{T\M} \defeq \{0_p \in T_p \M : p \in \M\}$ for the subset of $T\M$ containing the zero vector of each tangent space. A \textit{neighborhood} of $0_{T\M}$ is an open subset of $T\M$ containing $0_{T\M}$. For instance, a neighborhood $\W$ of $0_{T\R^n}$ could be specified by $\R^n \times W$, for $W$ a neighborhood of $0$ in $\R^n$.
\begin{theo}\label{theo:nonstr-D}
    Let $h$ be a CBF for $\Sigma = (\U, \M, F)$ on a neighborhood $\D$ with compact $\C$ and $\chi(\C) \neq 0$. Let $\kappa$ be a controller for which $p \mapsto F(p, \kappa(p))$ is continuous and satisfies \eqref{eq:nonstr-cbf-neigh}. For every neighborhood $\Vc$ of $\C$, there is a neighborhood $\W$ of $0_{T\M}$ such that for all $Z \in \X^0(\M)$ taking values in $\W$, there exist $(p, u) \in \Vc \times \U$ with
    \begin{align}
        F(p, u) = Z_p.
    \end{align}
\end{theo}
\begin{proof}
    Fix any neighborhood $\V$ of $\C$. Applying Lemma \ref{lem:c-perturb} with a neighborhood $\V \cap \D$ of $\C$, we conclude the existence of a compact regular domain $\tilde \C \subseteq \V \cap \D$ containing $\C$ for which the closed-loop vector field $p \mapsto F(p, \kappa(p))$ points strictly inward at $\partial \tilde \C$. As $\tilde \C$ is also diffeomorphic to $\C$, it is homotopy equivalent to $\C$; hence $\chi(\tilde \C) = \chi(\C) \neq 0$.

    Since $p \mapsto F(p, \kappa(p))$ is continuous and inward-pointing on $\partial \tilde \C$, there is a neighborhood $\W$ of $0_{T\M}$ such that, for any vector field $Z$ with $Z_p \in \W$ for all $p \in \M$, the perturbed vector field $p \mapsto F(p, \kappa(p)) - Z_p$ remains inward-pointing at $\partial \tilde \C$. Fix any continuous vector field $Z$ with $Z_p \in \W$ for all $p$. By Proposition \ref{prop:ec-zeros}, $p \mapsto F(p, \kappa(p)) - Z_p$ has a zero in $\tilde \C$. Thus, there are $(p, u) \in \tilde \C \times \U \subseteq \V \times \U$ with $F(p, u) = Z_p$.
\end{proof}
Theorem \ref{theo:nonstr-D} provides a necessary condition analogous to \cite[Thm. 3.2]{kvalheim2021necessary}. Unlike this result, however, Theorem \ref{theo:nonstr-D} leans on the geometry of $\C$ to avoid assuming unique integrability.

\subsection{Strict CBF}
Finally, we consider a system $\Sigma = (\U, \M, F)$ with a CBF $h$ with $\alpha \in \K_\infty^e$ and a controller $\kappa$ for which
\begin{align}
    dh_p F(p, \kappa(p)) > -\alpha(h(p)), \; \forall p \in \C. \label{eq:str-cbf}
\end{align}
\begin{theo}\label{theo:str-c}
    Let $h$ be a strict CBF for $\Sigma = (\U, \M, F)$ with compact $\C$ and $\chi(\C) \neq 0$. Let $\kappa$ be a controller for which $p \mapsto F(p, \kappa(p))$ is continuous and satisfies \eqref{eq:str-cbf}. Then, there is a neighborhood $\W$ of $0_{T\M}$ such that for all $Z \in \X^0(\M)$ taking values in $\W$, there exist $(p, u) \in \C \times \U$ with
    \begin{align}
        F(p, u) = Z_p.
    \end{align}
\end{theo}
\begin{proof}
    As $p \mapsto F(p, \kappa(p))$ is continuous and inward-pointing at $\partial \C$, there is a neighborhood $\W$ of $0_{T\M}$ for which $p \mapsto F(p, \kappa(p)) - Z_p$ is inward-pointing at $\partial \C$ for any $Z$ with $Z_p \in \W$ for all $p \in \M$. Fix any continuous vector field $Z$ with $Z_p \in \W$ for all $p \in \M$. By Proposition \ref{prop:ec-zeros}, $p \mapsto F(p, \kappa(p)) - Z_p$ has a zero in $\C$. Thus, there exists a pair $(p, u) \in \C \times \U$ satisfying $F(p, u) = Z_p$. 
\end{proof}
Theorem \ref{theo:str-c} recovers the necessary condition of \cite[Thm. 3.6]{kvalheim2021necessary} in the CBF setting. Here we again leverage the safe set geometry to avoid the unique integrability assumption of \cite{kvalheim2021necessary}.

\begin{cor}[Affine Strict CBF]\label{cor:aff-str-cbf}
    Let $h$ be a strict CBF for a continuous control-affine system $\Sigma = (\M, X_0, ..., X_m)$ with compact $\C$ and $\chi(\C) \neq 0$. Then, there is a neighborhood $\W$ of $0_{T\M}$ such that for all $Z \in \X^0(\M)$ taking values in $\W$, there exist $(p, u) \in \C \times \R^m$ with
    \begin{align}
        X_0|_p + \sum_{i = 1}^m u^i X_i|_p = Z_p.
    \end{align}
\end{cor}
\begin{proof}
    By Proposition \ref{prop:strict-aff}, there exists a smooth controller $\kappa: \M \to \R^m$ satisfying \eqref{eq:str-cbf} for all $p \in \C$. Thus, the conditions of Theorem \ref{theo:str-c} are satisfied, and the result follows.
\end{proof}

\subsection{Examples}
\begin{example}\label{ex:nonstr-cbf}
    Let $\U = \{u \in \R^2 : \norm{u}_{\ell^2} = 1\}$, and consider the system $\Sigma = (\U, \R^2, F)$ on $\R^2 \times \U$ defined by $\dot x = x + u$. We wish to find a CBF and a continuous controller for $\Sigma$ satisfying \eqref{eq:nonstr-C} for a safe set $\C = \{x \in \R^2 : x^\top x \leq 1\}$, which has $\chi(\C) = 1$. Although this task \textit{seems} feasible, as $h: x \mapsto 1 - x^\top x$ is a CBF on $\C$, we show it is impossible. The vector field $Z : x \mapsto (x, x) \in T\R^2$ is continuous and outward-pointing at $\partial \C$. Thus, it satisfies $dh_x Z_x \leq 0$ on $\partial \C$ for any candidate CBF $h$ for $\C$. However, there are no $(x, u) \in \C \times \U$ for which $x + u = x$, as this would require $u = 0 \notin \U$. In turn, there are no $(x, u) \in \C \times \U$ for which $F(x, u) = Z_x$. By Theorem \ref{theo:nonstr-cbf-c}, no pair of a CBF and controller with continuous closed-loop satisfying \eqref{eq:nonstr-C} for $\C$ can exist.
\end{example}
As in the case of stabilization \cite[Prop. 0]{pomet1992explicit}, \textit{kinematic nonholonomic systems} fail to meet the necessary conditions of Theorems \ref{theo:nonstr-D} and \ref{theo:str-c}. We consider one such example below.

\begin{example}
    Consider a kinematic model of an underactuated satellite on $SO(3)$, specified by $\dot R = u^1 R\hat e_1 + u^2 R \hat e_2$, for $e_1 = (1, 0, 0)$, $e_2 = (0, 1, 0)$, and $\wedge$ the hat map. Fix a regular domain $\C \subseteq SO(3)$ (which is automatically compact) with $\chi(\C) \neq 0$. Consider the vector field specified by $Z_\epsilon : R \mapsto \epsilon R \hat e_3$, for $e_3 = (0, 0, 1)$; for every neighborhood $\W$ of $0_{TSO(3)}$, there is an $\epsilon \neq 0$ for which $Z_\epsilon$ takes values in $\W$. However, there are no $(R, u) \in SO(3) \times \R^2$ for which $\epsilon R \hat e_3 = u^1 R \hat e_1 + u^2 R \hat e_2$. Corollary \ref{cor:aff-str-cbf} implies there is no strict CBF for $\Sigma$ with safe set $\C$. Similarly, Theorem \ref{theo:nonstr-D} implies there is no pair of a CBF on a neighborhood of $\C$ and corresponding safe and continuous controller $\kappa$ for $\C$.
\end{example}

\subsection{Comparison with Brockett's Necessary Condition}
\begin{table}[]
    \centering
    \caption{Comparison of Brockett \& CBF Conditions.}
    \begin{tabular}{|c|c|c|c|}
        \hline
        \; & $\begin{matrix}
            \text{Unique}\\
            \text{Integr.}
        \end{matrix}$  
        & $\begin{matrix}
            \text{Perturbation}
        \end{matrix}$ 
        & $\begin{matrix}
            \text{Necessary}\\
            \text{Condition}
        \end{matrix}$\\
        \hline
        \hline
        CBF & \xmark & $\{Z : dh_p Z_p \leq 0\}$ & $\begin{matrix}
            \exists (p, u) \in \C \times \U:\\
            F(p, u) = Z_p
        \end{matrix}$\\
        CBF on $\D$ & \xmark & $\{Z : Z_p \in \W \}$ & $\begin{matrix}
            \exists (p, u) \in \V \times \U:\\
            F(p, u) = Z_p
        \end{matrix}$\\
        Strict CBF & \xmark & $\{Z : Z_p \in \W \}$ & $\begin{matrix}
            \exists (p, u) \in \C \times \U:\\
            F(p, u) = Z_p
        \end{matrix}$\\
        Brockett & \checkmark & $z \in \W$ & $\begin{matrix}
            \exists (x, u) \in \R^n \times \R^m:\\
             f(x, u) = z
        \end{matrix}$\\
         \hline
    \end{tabular}
    \label{tab:brockett-CBF}
\end{table}
As demonstrated in Table \ref{tab:brockett-CBF}, the necessary conditions proposed by Theorems \ref{theo:nonstr-D} and \ref{theo:str-c} share remarkable structural similarities with Brockett's necessary condition.

These similarities are illuminated by the duality between Lyapunov and barrier functions. Consider a continuous, uniquely integrable vector field $\dot x = f(x)$ with asymptotically stable equilibrium point $x^*$. By the converse Lyapunov theorem, there is a smooth, proper Lyapunov function $V$ for $x^*$ satisfying $\dot V \leq -\alpha(V)$ on its domain of attraction \cite{wilson1969smoothing}. Each Lyapunov sublevel set $V^{-1}(\R_{\leq c})$ is a compact regular domain with nonzero Euler characteristic. Thus, for each $c > 0$ the function $h_c = c - V$ satisfies a barrier inequality in a neighborhood of $x^*$, and settings parallel to Theorems \ref{theo:nonstr-D} and \ref{theo:str-c} are recovered. Using this simple method to re-imagine Lyapunov sublevel sets as safe sets, the parallels between the necessary conditions become evident. 

The main differences between the necessary conditions---unique integrability of the closed loop \& use of geometric language---are explained as follows. The proof of Brockett's condition requires a converse Lyapunov theorem, which in turn requires unique integrability---since we \textit{do not} rely on converse results, we do without this assumption. The remaining differences between the necessary conditions (for example the change from $z$ to $Z$ and $\W$ to $\W$ containing $0_{T\M}$) are due to the transition between the Euclidean setting of Theorem \ref{theo:brockett} and geometric settings of Theorems \ref{theo:nonstr-D} and \ref{theo:str-c}.

\section{Conclusion}
In this work, we derived Brockett-like necessary conditions for the existence of CBFs and their safety-critical controllers. By utilizing the geometry of CBF safe sets, we provided streamlined derivations of these necessary conditions in three key cases: a non-strict CBF on $\C$, a non-strict CBF on a neighborhood of $\C$, and a strict CBF, all under the simple assumption of continuity of the closed-loop system. Future work includes applying these conditions to facilitate the efficient automated synthesis of CBFs.

\bibstyle{ieeetr}
\bibliography{references}

\end{document}